\documentclass[letterpaper,11pt]{article}
\usepackage[hmargin=1in,vmargin=1in]{geometry}
\usepackage{amsmath, amssymb}
\usepackage{bbm}
\usepackage{graphicx}
\usepackage{algorithm}
\usepackage{algorithmic}
\usepackage{xspace}
\title{Hardness of Vertex Deletion  and Project Scheduling\thanks{
This research was supported by Grant 228021-ECCSciEng of the European
Research Council.}}
\author{Ola Svensson \\
 EPFL, Switzerland \\
  {\tt ola.svensson@epfl.ch }}

\newtheorem{theorem}{Theorem}[section]

\newtheorem{observation}[theorem]{Observation}
\newtheorem{lemma}[theorem]{Lemma}
\newtheorem{claim}[theorem]{Claim}

\newtheorem{conjecture}[theorem]{Conjecture}

\newenvironment{proof}{\begin{trivlist}
    \item[\hskip\labelsep {\bf Proof}.]}{\QED \end{trivlist}}

\newcommand{\QED}{\hfill $\square$}

\newcommand{\hide}[1]{
}

\newcounter{myclaim}
\newcommand{\Inf}{\ensuremath{\mbox{\textnormal{Infl}}}\xspace}
\newcommand{\E}{\ensuremath{\mathbb{E}}\xspace}

\begin{document}
\maketitle

\begin{abstract}
Assuming the Unique Games Conjecture, we show strong inapproximability
results for two natural vertex deletion problems on directed graphs:
for any integer $k\geq 2$ and arbitrary small $\epsilon > 0$, the
Feedback Vertex Set problem and the DAG Vertex Deletion problem are
inapproximable within a factor $k-\epsilon$
even on graphs where the vertices can be almost partitioned into $k$
solutions. This gives a more structured and therefore stronger
UGC-based hardness result for the Feedback Vertex Set problem that is
also simpler (albeit using the ``It Ain't Over Till It's Over''
theorem) than the previous hardness result.

In comparison to the classical Feedback Vertex Set problem, the DAG
Vertex Deletion problem has received little attention and, although we
think it is a natural and interesting problem, the main motivation for
our inapproximability result stems from its relationship with the
classical Discrete Time-Cost Tradeoff Problem. More specifically, our
results imply that the deadline version is NP-hard to approximate
within any constant assuming the Unique Games Conjecture. This
explains the difficulty in obtaining good approximation algorithms for
that problem and further motivates previous alternative approaches
such as bicriteria approximations.

\end{abstract}

\section{Introduction}\label{sec:intro}
 
Many interesting problems can be formulated as that of finding a large induced subgraph satisfying a desired property of a given (directed) graph. One of the most well studied such problems is the \emph{Feedback Vertex Set (FVS)} problem where the property is acyclicity, i.e., given a directed graph $G= (V,E)$ we wish to delete the minimum number of vertices  so that the resulting graph is acyclic.  Another example is the \emph{DAG Vertex Deletion (DVD)} problem, 
where we are given an integer $k$ and a directed acyclic graph and we wish to delete the minimum number of vertices  so that 
the resulting graph has no path of length\footnote{For notational convenience, we shall measure the length of a path in terms of the number of vertices it contains instead of the number of edges.}  $k$.

The FVS problem and the related Feedback Arc Set problem was shown to be \textsc{NP}-complete already in Karp's seminal paper~\cite{Kar72} 
and there is a long history of approximation algorithms for these problems.  Leighton and Rao~\cite{LR88} first gave a $O(\log^2 |V|)$-approximation algorithm. Seymour~\cite{Seymour95} improved the approximation guarantee by showing that a certain linear program approximates the value within a factor $O(\log |V| \log \log |V|)$. Seymour's arguments were then generalized by Even et al.~\cite{ENSS98} to obtain the best known approximation algorithms  achieving a factor $O(\log |V| \log \log |V|)$ 
even in weighted graphs. 

Motivated by certain VLSI design and communication problems, Paik et al.~\cite{Paik94} considered the DVD problem and  showed it  to be \textsc{NP}-complete on general graphs and polynomial time solvable on series-parallel graphs.
One can also see that DVD for a fixed $k$ is a special case of the Vertex Cover problem on $k$-uniform hypergraphs  and has a fairly straightforward $k$-approximation algorithm.
 
 In comparison to FVS, the DVD problem has received little attention
 and, although we think it is a natural problem, our main motivation
 for studying its approximability comes from its relationship (that we
 prove in Section~\ref{sec:timecost}) with the classical deadline
 version of the project scheduling problem known as the \emph{Discrete
   Time-Cost Tradeoff problem}.  Informally (see
 Section~\ref{sec:timecost} for a formal definition of the Deadline
 problem), this is the problem where we are given a deadline and a
 project consisting of tasks related by precedence constraints, and the
 time it takes to execute each task depends, by a given cost function,
 on how much we pay for it. The objective is to minimize the cost of
 executing all the tasks in compliance with the precedence constraints
 so that they all finish within the given deadline. Due to its obvious
 practical relevance, the problem has been studied in various contexts
 over the last 50 years (see the paper~\cite{Kelley59} by Kelly and
 Walker for an early reference).  Fulkerson~\cite{Fulk61} and
 Kelley~\cite{Kelley61} obtained polynomial time algorithms if all
 cost functions are linear. In contrast, the problem becomes
 \textsc{NP}-hard for arbitrary cost functions~\cite{De95} and there
 is even no known constant factor approximation algorithm in the
 general case. However, better (approximation) algorithms have been
 obtained for special cases. For example, Grigoriev and
 Woeginger~\cite{GW04} gave polynomial time algorithms for special classes
 of precedence constraints and one of several algorithms by
 Skutella~\cite{Skutella98} is a bicriteria approximation that, for
 any $\mu \in (0,1)$, approximates the Deadline problem within a
 factor $1/(1-\mu)$ if the deadline is allowed to be violated by a
 factor $1/\mu$.

In summary, there are no known constant approximation algorithms for
FVS, DVD, and the Deadline problem although few strong
inapproximability results are known. The best known
\textsc{NP}-hardness of approximation results follow from the fact that they
are all as hard to approximate as Vertex Cover which is
\textsc{NP}-hard to approximate within a factor
$1.3606$~\cite{Dinur04onthe}.  It is indeed easy to see that Vertex
Cover is a special case of FVS and DVD, and Grigoriev and
Woeginger~\cite{GW04} gave an approximation-preserving reduction from
Vertex Cover to the Deadline problem.  If we assume the Unique Games
Conjecture (UGC)~\cite{Khot02a}, our understanding of the approximability of FVS
becomes significantly better: the hardness of approximation result for
Maximum Acyclic Subgraph by Guruswami et al.~\cite{GHMRC11} implies
that it is \textsc{NP}-hard to approximate FVS within any constant
factor assuming the UGC. However, the results in~\cite{GHMRC11} use
very sophisticated techniques that are not known to imply a similar
hardness for DVD and the Deadline problem.

Even though the starting motivation of this work was to better
understand the approximability of the Deadline problem (and DVD), the
techniques that we develop also lead to a stronger UGC-based hardness
result for FVS: similar to the recent results for Vertex Cover on
$k$-uniform hypergraphs by Bansal and Khot~\cite{BK09,Bansal2010}, we show that, for any
integer $k\geq 2$ and arbitrarily small $\epsilon > 0$, there is no
$k-\epsilon$-approximation algorithm for FVS even on graphs where the
vertices can be almost partitioned into $k$ feedback vertex sets. Our
reduction is also much simpler than the one in~\cite{GHMRC11} (albeit
using the ``It Ain't Over Till It's Over'' theorem) but is tailored for
FVS and does not yield any inapproximability result for the Maximum
Acyclic Subgraph problem. More importantly, our techniques  also lead to  an analogous result for the DVD problem (and thereby the Deadline problem).
Formally, our results for the considered vertex deletion problems can be stated as follows.
\newpage
\begin{theorem}
\label{thm:vertmain}
Assuming the Unique Games Conjecture, for any integer $k \geq 2$ and arbitrary constant  $\epsilon>0$, the following problems are \textsc{NP}-hard:

\begin{description}
\item[{FVS:}] Given  a graph $G(V,E)$, distinguish between the following cases:
 \begin{itemize}
    \item \emph{(Completeness):} there exist disjoint subsets $V_1, \dots, V_k \subset V$ satisfying $|V_i| \geq \frac{1-\epsilon}{k} |V|$ and such that a subgraph induced by all but one of these subsets is acyclic.
        \item  \emph{(Soundness):} every feedback vertex set has size at least  $(1-\epsilon)|V|$.
 \end{itemize}
\item[DVD:] Given a DAG $G(V,E)$, distinguish between the following cases:
 \begin{itemize}
    \item \emph{(Completeness):} there exist disjoint subsets $V_1, \dots, V_k \subset V$ satisfying $|V_i| \geq \frac{1-\epsilon}{k} |V|$ and such that a subgraph induced by all but one of these subsets has no path of length $k$.
        \item  \emph{(Soundness):} every induced subgraph of  $\epsilon |V|$ vertices has a path of length  $|V|^{1-\epsilon}$.
 \end{itemize}

\end{description}
\end{theorem}
Note that in the completeness cases, letting $V' = V \setminus (V_1 \cup \dots \cup V_k )$, the
sets $V' \cup V_i$ for  $i = 1, \dots , k$ are almost disjoint solutions
of size at most $( \frac{1}{k} + \epsilon)|V|$ each. In contrast, any solution basically needs to delete all vertices in the soundness case 
(even to avoid paths of length $|V|^{1-\epsilon}$ for DVD).

When proving UGC-based  inapproximability results, the main
task is usually to design ``gadgets'' of the considered problems that
simulate a so-called dictatorship test.  Once we have such
``dictatorship gadgets'', the process of obtaining  UGC-based
hardness results often follows from (by now) fairly standard
arguments. In particular, the main ideas needed for our reductions leading to
Theorem~\ref{thm:vertmain} are already present in the design of the gadgets.
We have therefore chosen to present those gadget constructions with
less cumbersome notation in 
Section~\ref{sec:dictator} and the reductions from Unique Games can
be found in Section~\ref{sec:ughardness}.

As alluded to above, our main interest in DVD stems from its relationship with the Deadline problem. More specifically, in Section~\ref{sec:timecost}, 
we give an approximation-preserving reduction from DVD to the Deadline problem that combined with Theorem~\ref{thm:vertmain} yields:
\begin{theorem}
Conditioned on the Unique Games Conjecture, for every $C>0$, it is \textsc{NP}-hard to find a $C$-approximation to the Deadline problem.
\end{theorem} 

This explains the difficulty in obtaining good approximation
algorithms for the Deadline problem and also further motivates
alternative approaches such as the bicriteria approach by
Skutella~\cite{Skutella98} that approximates the Deadline problem
within a constant if the deadline is allowed to be violated by a
constant factor.

\section{Preliminaries}
\label{sec:prelim}
\subsection{Low Degree Influence  and ``It Ain't Over Till It's Over'' Theorem}

Let $[k] = \{0,1, \dots, k-1\}$. When analyzing our hardness reductions, we shall use known properties
regarding the behavior of functions of the form $f: [k]^R \mapsto
\{0,1\}$ depending on whether they have influential co-ordinates. Similar to~\cite[Section~$3$]{MOO10}, we define the influence of the $i$-th co-ordinate  by
$$
\Inf_i(f) = \E_x [ \mbox{Var}(f)| x_1, \dots, x_{i-1}, x_{i+1}, \dots, x_R].
$$ We note that if $f: \{-1,1\}^R \mapsto \{-1,1\}$ then this definition coincides with the intuitive expression  $\Pr_x[ f(x_1,\dots, x_i, \dots, x_R) \neq f(x_1, \dots, -x_i, \dots, x_R)]$.

It is well known that if we let $f = \sum_\Phi \hat f(\phi) X_\phi$ be the multi-linear representation of $f$ (where, analogous of the
standard Fourier representation, the characters $(X_\phi)_{\phi\in
  [k]^R}$ define an orthonormal basis of the vector space of all functions $[k]^n \mapsto \mathbb{R}$) then  the influence can also be expressed as
$$
\Inf_i(f) = \sum_{\phi: \phi_i \neq 0} \hat f^2(\phi),
$$ which motivates the following definition of the degree
$d$-influence of the $i$-th co-ordinate:
$$
\Inf_i^d(f) = \sum_{\phi: \phi_i \neq 0, |\phi|\leq d} \hat f^2(\phi).
$$ As we shall not work directly with these definitions or with the
multi-linear representation, we refer the reader to~\cite{MOO10} for
the precise definitions and  cut the discussion short by mentioning the
property of low degree influence that shall be crucial to us (which
follows from that $\sum_{\phi}\hat f^2(\phi) = \E_x[f(x)^2] \leq 1$).

\begin{observation}
\label{obs:lowdeg}
For a boolean function $f: \{0,1\}^R \mapsto \{0,1\}$,
the sum of all degree $d$-influences is at most $d$.
\end{observation}

We shall now introduce a simplified version of the ``It Ain't Over
Till It's Over'' theorem that is sufficient for the applications in
this paper. The first proof was given  by Mossel et al.~\cite{MOO10}
and a more combinatorial proof of a simplified version (very similar
to the one used here) was given  by Bansal and Khot~\cite{BK09} who
used it to prove tight inapproximability results for Vertex Cover and a
classical single machine scheduling problem. In fact many of our ideas are inspired from~\cite{BK09}.   For $x\in [k]^R$ and a subsequence $S_\epsilon=(i_1, \dots,
   i_{\epsilon R})$ of $\epsilon R$ not necessarily distinct indexes
   in $[R]$, let
$$
C_{x,S_\epsilon} =  \{z \in [k]^R: z_j = x_j \mbox{ }\forall j\not \in S_\epsilon\}
$$ denote the sub-cube defined by fixing the co-ordinates not in
$S_\epsilon$ according to $x$. Let also $f(C_{x,S_\epsilon}) \equiv 0$
denote the expression that $f$ is identical to $0$ on the sub-cube
$C_{x,S_\epsilon}$.

\begin{theorem}
\label{thm:soundsingle}
For every $\epsilon, \delta >0$ and integer $k$, there exists $\eta >0$ and integer
$d$ such that any $f: [k]^R \mapsto \{0,1\}$ that satisfies
$$
\E[f] \geq \delta  \qquad \mbox{and}  \qquad \forall i \in [R], \Inf_{i}^d(f) \leq \eta,
$$
has
$$
\Pr_{x,S_\epsilon}\left[ f(C_{x,S_\epsilon}) \equiv 0\right] \leq \delta.
$$
\end{theorem}
Here and throughout the paper, the probability over $x,S_\epsilon$ is
such that $x$ and $S_\epsilon$ are taken independently and uniformly
at random. When $\epsilon$ is clear from the context we often also abbreviate $S_\epsilon$ by  $S$. Note that the theorem says that a reasonably balanced
function with no low degree influential co-ordinates has very low
probability to be identical to $0$ over the random choice of
sub-cubes. In contrast, it is easy to see that a dictatorship function (on the boolean domain)
$f(x) = x_s$, for some $s$, has $\Pr_{x,S_\epsilon}\left[
  f(C_{x,S_\epsilon}) \equiv 0\right] = \Pr_{x,S_\epsilon}\left[
  f(C_{x,S_\epsilon}) \equiv 1\right] \geq 1/2-\epsilon$. It is this
drastic difference that we will exploit in our hardness reductions.

\subsection{Unique Games Conjecture}
 An instance of Unique Games $\mathcal{L} = (G(V,W,E),[R],
 \{\pi_{v,w}\}_{(v,w)})$ consists of a regular bipartite graph
 $G(V,W,E)$ and a set $[R]$ of labels. For each edge $(v,w) \in E$ there is a constraint
 specified by a permutation $\pi_{v,w} : [R] \mapsto [R]$. The goal is
 to find a labeling $\rho : (V\cup W) \mapsto [R]$ so as to maximize
 $val(\rho) := \Pr_{e\in E} [ \rho \mbox{ satisfies } e]$, where a
 labeling $\rho$ is said to satisfy an edge $e=(v,w)$ if $\rho(v) =
 \pi_{v,w}(\rho(w))$. For a Unique Games instance $\mathcal{L}$, we let
 $OPT(\mathcal{L}) = \max_{\rho: V\cup W \mapsto [R]} val(\rho)$. The
 now famous Unique Games Conjecture that has been extensively used to
 prove strong hardness of approximation results can be stated as
 follows.

\begin{conjecture}[\cite{Khot02a}]
\label{conj:ugc}
  For any constants $\zeta, \gamma >0$, there is a sufficiently large
  integer $R= R(\zeta, \gamma)$ such that, for  Unique Games instances
   $\mathcal{L}$ with label set $[R]$ it is NP-hard to distinguish between:
\begin{itemize}
\item \emph{(Completeness):} $OPT(\mathcal{L})\geq 1-\zeta$.
\item \emph{(Soundness):} $OPT(\mathcal{L}) \leq
  \gamma$.
\end{itemize}
\end{conjecture}

\section{Dictatorship Gadgets for Vertex Deletion Problems}
\label{sec:dictator}
We give fairly simple gadgets of the considered vertex deletion
problems that informally corresponds to a dictatorship test in the
following sense: (Completeness:) any dictatorship function $f:[k]^R
\mapsto [k]$ (defined by $f(x) = x_s$ for some $s\in [R]$) corresponds
to a good solution whereas (Soundness:) any non-trivial solution
corresponds to a function $f:[k]^R \mapsto \{0,1\}$ with a high
influence co-ordinate. By fairly standard arguments, these gadgets are then used in Section~\ref{sec:ughardness}  to obtain analogous hardness results assuming
the Unique Games Conjecture.

Throughout this section, we fix  $k$ to be an integer, $\epsilon, \delta >0$ to be arbitrarily small constants, and let $\eta$ and $d$ be as in Theorem~\ref{thm:soundsingle} (depending on $k,\epsilon$ and $\delta$).

\subsection{Feedback Vertex Set}
\label{sec:hardfvs}
We shall here describe a graph $G= (V,E)$ that naturally corresponds to a dictatorship test in the following sense:
 \begin{itemize}
 \item \emph{(Completeness:)} A dictatorship function partitions the vertex set into subsets
   $V', V_0, \dots, V_{k-1}$ satisfying $V_j \geq \frac{1-\epsilon }{k}
   |V|, |V'| \leq \epsilon  |V|$, and for $j\in [k]$ the graph obtained by deleting  $V' \cup V_j$  is acyclic.
\item  \emph{(Soundness:)} Any feedback vertex set that deletes less than  $(1-2\delta)|V|$ vertices corresponds to a  function $f:[k]^R \mapsto \{0,1\}$ with a co-ordinate $i$ so that $\Inf_i^d(f) > \eta$.
\end{itemize}

\subsubsection{Dictatorship Gadget}
To make the analysis more intuitive, it will be convenient to first present
a gadget that consists of two types of vertices that we refer to as
\emph{bit-vertices} and \emph{test-vertices} and all arcs are between
bit- and test-vertices:
\begin{itemize}
\item There is a bit-vertex  $b_x$ of weight $\infty$ for every $x\in   [k]^R$. 

\item There is a test-vertex $t_{x,S}$ of weight $1$ for every $x\in [k]^{R}$
  and every sequence  $S= (i_1, \dots, i_{\epsilon R}) \in [R]^{\epsilon  R}$ of $\epsilon R$ not necessarily distinct indices.

\item  The arc incident to a test-vertex $t_{x,S}$ are the
  following. There is an arc $(b_z, t_{x,S})$ if $z\in C_{x,S}$ and an arc $(t_{x,S}, b_z)$ if $z \in C^\oplus_{x,S}$, where 
  $$
  C^\oplus_{x,S} =  \{z \oplus 1: z\in C_{x,S}\}
  $$
  (here $\oplus$ denotes addition mod $k$).
\end{itemize}

As the bit-vertices have weight $\infty$, they will never be deleted
in an optimal solution. We can therefore obtain an unweighted graph
$G$ of same optimal value by omitting the bit-vertices and having an
arc $(t_{x,S}, t_{x',S'})$ between two test vertices if there exists a
bit-vertex $b_z$ so that $(t_{x,S}, b_z)$ and $(b_z, t_{x',S'})$.  The
vertex set of $G$ will therefore correspond to the set $T$ of
test-vertices. The analysis of $G$ therefore follows from proving that
(completeness:) any dictatorship function partitions the test-vertices
as required (Section~\ref{sec:fvsgadcomp}) and (soundness:) that any
solution that deletes less than a fraction $1-2\delta$ of the
test-vertices corresponds to a function with a co-ordinate of high
influence (Section~\ref{sec:fvsgadsound}).

\subsubsection{Completeness}
\label{sec:fvsgadcomp}
We show that a dictatorship function $f: [k]^R \mapsto [k]$ of index
$s$ naturally partitions the test-vertices into subsets $T', T_0,
\dots, T_{k-1}$ satisfying $T_j \geq \frac{1-\epsilon }{k} |T|, |T'|
\leq \epsilon  |T|$, and such that the sets $T' \cup T_j$ for $j\in [k]$ are
almost disjoint feedback vertex sets of size at most $(\frac{1}{k} +
\epsilon )|T|$ each.

As $f(x) = x_s$, it partitions the bit-vertices in $k$ equal
sized sets
$$
B_j = \{b_x: f(x) = j\} \qquad \mbox{for} \qquad j\in [k].
$$

We say that a test-vertex $t_{x,S}$ is good if $s\not \in S$ and partition
the good test-vertices into $k$ equal sized sets
$$
T_j = \{t_{x,S} : s\not \in S \mbox{ and } f(x) = j\} \qquad \mbox{ for } \qquad j\in [k].
$$ The sets are of equal size since they are partitioned according to
$x$ and whether a test-vertex is good only depends on
$S$. Furthermore, as at least a fraction $1-\epsilon$ of the
test-vertices are good we have that $|T_j| \geq \frac{1-\epsilon }{k}
|T|$ for $j\in [k]$ and therefore the remaining test-vertices in $T'$
are at most $\epsilon  |T|$ many.

It remains to show that $T_j \cup T'$ defines a feedback vertex set
for any $j\in [k]$. The key observation is that $T_j$ only have
incoming edges from bit-vertices in $B_j$ and outgoing edges to
bit-vertices in $B_{j\oplus 1}$. Indeed, consider a test-vertex
$t_{x,S}\in T_j$ and an arc $ (b_z, t_{x,S})$. By definition we have
that $z \in C_{x,S}$ and as $S$ is good we have that $f(z) = f(x)
= j$, which implies that $z \in B_j$. The exact same argument implies
that $t_{x,S}$ only has outgoing edges to $B_{j\oplus 1}$.

The graph obtained by deleting all bad test-vertices and one of the
sets $T_0, T_1, \dots, T_{Q-1}$ is therefore acyclic as required.

\subsubsection{Soundness}
\label{sec:fvsgadsound}

Let $A$ be the last $1/2$ fraction of the bit-vertices according to a
topological sort of the graph.  Let $f_A$ be the indicator function of
$A$. Note that a test-vertex $t_{x,S}$ has incoming arcs from all
bit-vertices in $C_{x,S}$ and outgoing arcs to all bit-vertices in
$C^{\oplus}_{x, S}$. Therefore, if a test-vertex $t_{x,S}$ is not
deleted then we must have that either $f_A$ is identical to $0$ on
$C_{x,S}$ (if $t_{x,S}$ is placed before the last bit-vertex for which
$f_A$ evaluates to $0$) or identical to $1$ on $C^\oplus_{x,S}$ (if
$t_{x,S}$ is placed after the last bit-vertex for which $f_A$
evaluates to $0$) depending on where $t_{x,S}$ is placed according to
the topological sort.

As $\E[f_A] = 1/2$, we have by Theorem~\ref{thm:soundsingle} that if $\Inf_i^d(f_A) \leq \eta$ for all $i \in [R]$ then
$$
\Pr_{x,S}[ f(C_{x,S}) \equiv 0] \leq \delta
$$
and
$$
\Pr_{x,S} [ f(C^\oplus_{x,S}) \equiv 1] = \Pr_{x,S} [ f(C_{x,S}) \equiv 1] =\Pr_{x,S} [(1-f)(C_{x,S}) \equiv 0]  \leq  \delta.
$$

Therefore, if the solution does not correspond to a function with a co-ordinate of high  low-degree influence it must have deleted at least a fraction $1-2\delta$ of the test-vertices.

\subsection{Dag Vertex Deletion Problem}
\label{sec:harddvdp}

 We shall describe a directed acyclic graph
 (DAG) $G=(V,E)$ that naturally corresponds to dictatorship test in
 the following sense: 
 \begin{itemize}
 \item \emph{(Completeness:)} A dictatorship function partitions the vertex set into subsets
   $V', V_0, \dots, V_{k-1}$ satisfying $V_j \geq \frac{1-\epsilon }{k}
   |V|, |V'| \leq \epsilon  |V|$, and such that for $j\in [k]$ the graph obtained by deleting  $V' \cup V_j$  has no path of length $k$.
 \item \emph{(Soundness:)} Any graph obtained by deleting less than $(1-6\delta)|V|$  vertices either has a path of length $|V|^{1-\delta}$ or  corresponds to a function  $f:
   [k]^R \mapsto \{0,1\}$ with a co-ordinate $i$ such that
   $\Inf_i^d(f) > \eta$.
 \end{itemize}

\subsubsection{Dictatorship Gadget}
 As in Section~\ref{sec:hardfvs}, it will be convenient to first
 present a gadget that consists of two types of vertices that we refer
 to as \emph{bit-vertices} and \emph{test-vertices}, and all edges
 will be between bit- and test-vertices:
\begin{itemize}
\item The bit-vertices are partitioned into $L+1$ bit-layers ($L$ is selected below). Each bit-layer $\ell=0, \dots, L$ contains a bit-vertex $b^{\ell}_{x}$ of weight $\infty$ for every $x\in [k]^R$.

\item Similarly, the test-vertices are partitioned into $L$
  test-layers.  Each test-layer $\ell = 0, \dots, L-1$ has a
  test-vertex $t^\ell_{x,S}$ of weight $1$ for every $x\in [k]^{R}$
  and every sequence of indices $S= (i_1, \dots, i_{\epsilon R}) \in [R]^{\epsilon R}$.

\item The arcs are the following: there is an arc $(b^\ell_z, t^{\ell'}_{x,S})$ if $\ell\leq \ell'$ and $z\in C_{x,S}$, and there is an arc $(t^{\ell'}_{x,S}, b^\ell_z)$ if $\ell>\ell'$ and $z\in C^\oplus_{x,S}$.
\item Finally, $L$ is selected so as $\delta  L \geq |T|^{1-\delta}$, where $T$ is the set of test-vertices.
\end{itemize}

Note that, as there are only arcs from a bit-layer $\ell$ to a
test-layer $\ell'$ if $\ell' \geq \ell$ and only arcs from a
test-layer $\ell'$ to a bit-layer $\ell$ if $\ell> \ell'$, the
constructed graph is acyclic. Similar to the gadget for FVS, the
bit-vertices can be omitted to obtain an unweighted graph $G$ (with the set $T$ of test-vertices as vertices) with the
same optimal value by having an arc between two test-vertices if there
was a path between them through one bit-vertex. Note that a path in $G$ of length $k$ is a path in the gadget that consists of $k$ test-vertices. When arguing about the gadget, we will therefore say
that a \emph{path has length $k$ if it consists of $k$ test-vertices}.

Similarly to Section~\ref{sec:hardfvs}, the analysis of $G$ follows
from proving that (completeness:) any dictatorship function partitions
the test-vertices as required (Section~\ref{sec:fvsgadcomp}) and
(soundness:) that any solution that deletes less than a fraction
$1-6\delta$ of the test-vertices either has a path of length
$|T|^{1-\delta}$ or corresponds to a function with a co-ordinate of
high influence (Section~\ref{sec:fvsgadsound}).

\subsubsection{Completeness}
We show that a dictatorship function $f: [k]^R \mapsto [k]$ of
index $s$ naturally partitions the test-vertices into subsets    $T', T_0, \dots, T_{k-1}$ satisfying $T_j \geq \frac{1-\epsilon }{k}
   |T|, |T'| \leq \epsilon  |T|$, and such that for $j\in [k]$ the graph obtained by deleting  $T' \cup T_j$  has no path of length $k$.

This can be seen by the same arguments as in
Section~\ref{sec:fvsgadcomp}. Indeed if we ``collapse'' the different
layers by identifying the different copies of bit- and test-vertices
then the gadget constructed here is identical to the gagdet in that
section. We can therefore (by the arguments of
Section~\ref{sec:fvsgadcomp}), partition the bit-vertices into $k$
equal sized sets $B_0, B_1, \dots, B_{k-1}$ and all but an $\epsilon
$ fraction of the test-vertices into $k$ equal sized sets $T_0, T_1,
\dots, T_{k-1}$ so that any test-vertex in $T_j$ has only incoming
arcs from bit-vertices in $B_j$ and outgoing arcs to bit-vertices in
$B_{j\oplus 1}$.

Any $j\in [k]$ therefore corresponds to a solution by removing an
$\epsilon$ fraction of the test-vertices (i.e., the set $T'$) and those test-vertices in
$T_j$.

\subsubsection{Soundness}
\label{sec:dvdpsound}
Before proceeding to the analysis it will be convenient to consider a
different but equivalent formulation of the problem.

First, note that in any solution to DVD, i.e., a subgraph so that each
path contains less than $k$ test-vertices, we can find a coloring
$\chi$ (using for example depth-first search) that assigns a color in
$\{1,2, \dots, k\}$ to the  bit-vertices with the property that,
for each remaining test-vertex, the maximum color assigned to its
predecessors is strictly less than the minimum color assigned to its
successors. Similarly, any such coloring $\chi$ can be turned into a
solution to DVD by deleting those test-vertices, for which not all
predecessors are assigned lower colors than all its
successors. Furthermore, from the construction of the arcs, we can
assume w.l.o.g that the coloring satisfies $\chi(b^\ell_x) \leq
\chi(b^{\ell'}_x)$ if $\ell \leq \ell'$.

From the above discussion, an equivalent formulation of DVD on the constructed instances is  as follows:
find a
coloring $\chi$ that assigns a color in $\{1,2,\dots, k\}$ to each
bit-vertex satisfying $\chi(b^\ell_x) \leq \chi(b^{\ell'}_x)$ if $\ell\leq \ell'$ so as to minimize the number of unsatisfied test-vertices
where a test-vertex $t^\ell_{x,S}$ is said to be satisfied if
$$
\max_{\substack{z\in C_{x,S}}} \chi(b^{\ell}_{z}) < \min_{\substack{z\in C^\oplus_{x,S}}} \chi(b^{\ell+1}_{z}),
$$
that is all its predecessors are assigned lower colors than its successors.

It will also be convenient to consider the following  lower bound on the colors assigned to most bit-vertices in each layer:
define the color $\chi(\ell)$ of a bit-layer $\ell = 0, 1,
 \dots, L$ as the maximum color that satisfies $ \Pr_{x}[ \chi(b^\ell_x)
   \geq \chi(\ell)] \geq 1-\delta. $

Now, with each test-layer $\ell= 0, 1, \dots, L-1$ we associate the
indicator function $f^\ell:[k]^R \mapsto \{0,1\}$ defined as follows
$$
f^\ell(x) =
\begin{cases}
0 & \mbox{if} \qquad \chi(b^{\ell+1}_x)  > \chi(\ell),  \\
1 & \mbox{otherwise}.
\end{cases}
$$

The key observation for the soundness analysis is the following.
\begin{claim}
\label{claim:keydictdvdp}
For $\ell=0, \dots, L-1$, assuming that $\Inf_i^d(f^\ell) \leq \eta$ for all $i\in[R]$: if a fraction $3\delta$ of the test-vertices of test-layer
$\ell$ are satisfied, then  $\chi(\ell) < \chi(\ell+1)$.
\end{claim}
\begin{proof}
As at least a fraction $3\delta$ of the test-vertices of test-layer $\ell$ are satisfied,
$$
\Pr_{x,S}\left[
\max_{z\in C_{x,S}} \chi(b^{\ell}_{z}) < \min_{z\in C^\oplus_{x,S}} \chi(b^{\ell+1}_{z})
\right] \geq 3\delta.
$$
By the definition of $\chi(\ell)$ we have  $\Pr_x[\chi(b_x^\ell) \geq
  \chi(\ell)] \geq 1-\delta$ and therefore
$$
\Pr_{x,S}\left[
\chi(\ell) <  \min_{z\in C^\oplus_{x,S}} \chi(b^{\ell+1}_{z})
\right] = \Pr_{x,S}\left[ f^\ell(C_{x,S}) \equiv 0 \right] \geq 2\delta.
$$
As $\Inf_i^d(f^\ell) \leq \eta$ for all $i\in [R]$,
 Theorem~\ref{thm:soundsingle} implies that $E[f^\ell] < \delta$ and hence $\chi(\ell+1) > \chi(\ell)$.
\end{proof}

If a coloring satisfies more than a fraction $6\delta$ of the
test-vertices then at least a $3\delta$ fraction of the test-layers
are such that at least a fraction $3\delta$ of the test-vertices of
that layer are satisfied, which in turn by the preceding claim implies
that either one of them corresponds to a function with a co-ordinate
of high influence or $3\delta L$ many colors are needed (or
equivalently the graph contains a path consisting of at least $3\delta
L-1 \geq \delta L \geq |T|^{1-\delta}$ test-vertices).

\section{Hardness Assuming the Unique Games Conjecture}
\label{sec:ughardness}

In order to turn our dictatorship gadgets into hardness proofs (assuming the Unique Games Conjecture), we need a more general "It Ain't Over Till It's Over" theorem that not only verifies that a given number of functions all are dictatorships but ideally they should also be dictators of the \emph{same} co-ordinate.
Again an even more general variant of the theorem follows from~\cite{MOO10} and an easier proof of a very similar version to the case presented here can be found in~\cite{BK09}.

\begin{theorem}
\label{thm:soundmultiple}
For every $\epsilon, \delta >0$ and integer $k$, there exists $\eta >0$ and integers
$t,d$ such that any collection of functions $f_1, \dots, f_t: [k]^R \mapsto \{0,1\}$ that satisfies
$$
\forall j, \E[f_{j}] \geq \delta  \qquad \mbox{and} \qquad   \forall i \in [R], \forall 1 \leq \ell_1 \neq \ell_2 \leq t, \min\left\{\Inf_{i}^d(f_{\ell_1}),  \Inf_{i}^d(f_{\ell_2})\right\} \leq \eta,
$$
has
$$
\Pr_{x,S_\epsilon}\left[ \bigwedge_{j=1}^t f_{j}(C_{x,S_\epsilon}) \equiv 0\right] \leq \delta.
$$
\end{theorem}

For the applications of this paper, the interesting implication of the above theorem can be formulated as follows: if $\Pr_{x,S_\epsilon}\left[ \bigwedge_{j=1}^t f_{j}(C_{x,S_\epsilon}) \equiv 0\right] > \delta$ for $t$ fairly balanced functions then at least two of them must have  a common influential co-ordinate.

In our (soundness) analyses, we associate a boolean function $f_{w}: \{0,1\}^R \mapsto \{0,1\}$ with each $w\in W$ of the considered Unique Games instance (in fact we shall, as in the dictatorship gadgets, use the domain $[k]^R$ but for simplicity we restrict this discussion to the binary case). We then use the preceding theorem to test whether  $t$ (or more) neighbors $w_1, \dots w_t$
of a vertex $v\in V$ are "close to" consistent, i.e., $f_{w_1}, \dots, f_{w_t}$ are dictatorships on co-ordinates $\rho(w_1), \dots, \rho(w_t)$ such that $\pi_{v,w_1}(\rho(w_1)) = \pi_{v,w_2}(\rho(w_2))= \dots = \pi_{v,w_t}(\rho(w_t))$.
Indeed, on the one hand, if they are consistent then it is easy to see that
$$
\Pr_{x,S_\epsilon}\left[ \bigwedge_{j=1}^t f_{w_j} \circ \pi_{v,w_j} (C_{x,S_\epsilon}) \equiv 0\right] \geq 1/2 -\epsilon,
$$
and on the other hand (assuming $f_{w_1}, \dots, f_{w_t}$ are fairly balanced) if no two of them are "close to" consistent then Theorem~\ref{thm:soundmultiple} implies that
$$
\Pr_{x,S_\epsilon}\left[ \bigwedge_{j=1}^t f_{w_j} \circ \pi_{v,w_j} (C_{x,S_\epsilon}) \equiv 0\right] \leq \delta,
$$
where $f_{w_j} \circ \pi_{v,w_j}(x) = f_{w_j}( x_{\pi_{v,w_j}(1)}, x_{\pi_{v,w_j}(2)}, \dots, x_{\pi_{v,w_j}(R)})$. Similar to the gadget reductions, it is this drastic difference that we exploit to obtain our hardness results.

For the reductions, it shall be convenient to let $C_{x,S, v, w}$ denote the sub-cube
$$
C_{x,S, v, w}  = \{z: z_j = x_{\pi_{v,w}(j)}\ \forall j: \pi_{v,w}(j)\not \in S\}, 
$$
i.e., the image of the sub-cube $C_{x,S}$ via $\pi_{v,w}$.  Note that with this notation we have that
$$
\Pr_{x,S}\left[ \bigwedge_{j=1}^t f_{w_j} \circ \pi_{v,w_j} (C_{x,S}) \equiv 0\right] = \Pr_{x,S}\left[ \bigwedge_{j=1}^t  f_{w_j}(C_{x,S,v,w_j}) \equiv 0
\right]
$$

We now present the adaptations of the dictatorship gadgets for FVS and DVDP to obtain reductions from Unique Games in Sections~\ref{sec:ugcfvs} and~\ref{sec:ugcdvdp}, respectively.
Throughout this section (as in Section~\ref{sec:dictator}), we fix $k$ to be an integer, $\epsilon, \delta >0$ to be arbitrarily small constants and let $\eta, d,$ and $t$ be as in Theorem~\ref{thm:soundmultiple} (depending on $k, \epsilon$ and $\delta$).

\subsection{Feedback Vertex Set}
\label{sec:ugcfvs}
We prove the following theorem which clearly implies the FVS hardness stated in Theorem~\ref{thm:vertmain}.

\begin{theorem}
\label{thm:ugcfvs}
Assuming the Unique Games Conjecture, for any integer $k \geq 2$ and arbitrary constants  $\epsilon, \delta>0$, given a directed graph $G(V,E)$, distinguishing between the following cases is \textsc{NP}-hard:
 \begin{itemize}
    \item \emph{(Completeness):} there exist disjoint subsets $V_1, \dots, V_k \subset V$ satisfying $|V_i| \geq \frac{1-2\epsilon}{k} |V|$ and such that a subgraph induced by all but one of these subsets is acyclic.
        \item  \emph{(Soundness):} every induced subgraph of $8\delta |V|$ vertices contains a cycle.
 \end{itemize}
\end{theorem}

We first present the reduction in the following subsection followed by the completeness (Lemma~\ref{lemma:ugcfvscomp}) and soundness (Lemma~\ref{lemma:ugcfvssound}) analyses.

\subsubsection{Reduction}

We  describe a reduction from Unique Games to FVS. Let $\mathcal{L}(G(V,W,E), [R], \{\pi_{v,w}\}_{(v,w)\in E}$ be  a Unique Games instance. As in Section~\ref{sec:hardfvs},  the FVS instance consists of two types of vertices that we refer to as \emph{bit-vertices} and \emph{test-vertices} and all edges are between bit- and test-vertices.
\begin{itemize}
\item For every $w\in W$ and $x \in [k]^{R}$, there is a bit-vertex $b_{w,x}$ of weight $\infty$.

    In other words, each $w\in W$ is replaced by a $k$-ary hypercube $[k]^R$ where each vertex has weight $\infty$ so that none of them will ever be deleted in an optimal solution.

\item For every $v\in V, (w_1, \dots, w_{2t}) \in N(v)^{2t}, x\in [k]^R$ and $S= (i_1, i_2, \dots, i_{\epsilon R}) \in [R]^{\epsilon R}$,  we have a test vertex $t_{x,S, v, w_1, \dots, w_{2t}}$.

\item 
    The arcs incident to a test-vertex $t_{ x,S,v, w_1, \dots, w_{2t}}$ are the following. For $j = 1, \dots, 2t$,
    \begin{itemize}
    \item there is an arc  $(b_{w_j, z},t_{v,x,S, w_1, \dots, w_{2t}})$ if $z\in C_{x,S,v,w_j}$,
    \item   and an arc $( t_{v,x,S, w_1, \dots, w_{2t}},b_{w_j,z})$ if $z\in C^{\oplus}_{x,S,v,w_j}= \{z \oplus 1: z\in C_{x,S,v,w_j}\}$.
    \end{itemize}
\end{itemize}

As the bit-vertices have weight $\infty$, they are never deleted in an optimal solution and we can obtain
 an unweighted graph $G$ (with the set $T$ of test-vertices as vertices) with the
same optimal value by having an arc between two test-vertices if there
is a path between them through one bit-vertex.
Theorem~\ref{thm:ugcfvs} therefore follows from proving that (i) we can partition the test-vertices as required in the completeness case (Lemma~\ref{lemma:ugcfvscomp}) and  (ii) that we have to delete almost all test-vertices in the soundness case (Lemma~\ref{lemma:ugcfvssound}).

\subsubsection{Completeness}
\label{sec:ugcfvscomp}

We prove the following.
\begin{lemma}
\label{lemma:ugcfvscomp}
If there is a labeling $\rho$ of the Unique Games
instance $\mathcal{L}$ satisfying a $1-\zeta$ fraction of the constraints then we can partition the test-vertices into subsets $T', T_0, \dots, T_k$ satisfying
$T_j \geq \frac{1-2\epsilon }{k}
   |T|, |T'| \leq 2\epsilon  |T|$, and for $j\in [k]$ the graph obtained by deleting  $T' \cup T_j$  is acyclic.
\end{lemma}
\begin{proof}

Let $\rho$ be such a labeling of the Unique Games instance. We now use $\rho$ to partition the bit-vertices in $k$ equal sized sets:
$$
B_j = \{b_{w,x}: w\in W\mbox{ and } x_{\rho(w)} = j\} \qquad \mbox{ for } j\in [k].
$$

We say that a test-vertex $t_{x,S,v,w_1, \dots, w_{2t}}$ is good if (i)
$\rho(v) \not \in S$ and (ii) $\rho$ satisfies all the edges
$(v,w_1), (v,w_2), \dots, (v, w_{2t})$. Note that property (i) holds with
probability at least $1-\epsilon $ and property (ii) holds with
probability at least $1-\zeta 2t$. Therefore, at least a fraction of $
1- 2\epsilon $ (for $\zeta$ small enough) of the test-vertices are
good. As we did for the bit-vertices, we partition the test-vertices into $k$
equal sized sets:
$$
T_j = \{t_{x,S,v,w_1, \dots, w_{2t}}: t_{x,S,v,w_1, \dots, w_{2t}} \mbox{ is good  and } x_{\rho(v)}  = j \} \qquad \mbox{ for } j \in [k].
$$
The sets are of equal size since they are partitioned
according to $x$ and whether a test-vertex is good only depends on $S$
and $v,w_1, \dots, w_{2t}$. Furthermore, since at least $(1-2\epsilon)|T|$ test-vertices are good we have that $|T_j| \geq \frac{1-2\epsilon}{k} |T|$ for $j\in[k]$ and the remaining test-vertices in $T'$ are therefore at most $2\epsilon |T|$ many.

It remains to show that $T_j \cup T'$ defines a feedback vertex set for any $j\in [k]$. The key observation is that $T_j$ only have incoming arcs from bit-vertices in $B_j$ and outgoing arcs to bit-vertices in $B_{j\oplus 1}$. To see this,
consider a test-vertex $t_{x,S, v,w_1, \dots, w_{2t}}$ in $T_j$ and let $(b_{w_i,z}, t_{x,S, v,w_1, \dots, w_{2t}})$ be an arc. Then, by definition we have that $z \in C_{x,S,v, w_i}$.  As $\pi_{v,w_i}(\rho(w_i)) = \rho(v) \not \in S$, 
$$z_{\rho(w_i)} =  x_{\pi_{v,w_i}(\rho(w_i))} = x_{\rho(v)} = j,$$ and hence $b_{w_i,z} \in B_j$.
The exact same
argument also shows that all outgoing arcs from $t_{x,S, v, w_1, \dots, w_{2t}}$ goes to bit-vertices in
$B_{j\oplus 1}$. We can therefore conclude that test-vertices in $T_j$ have
only incoming arcs from bit-vertices in $B_j$ and outgoing arcs to
bit-vertices in $B_{j\oplus 1}$.

By the key observation, we can obtain an acyclic graph by deleting all bad test-vertices and one of the sets $T_0, \dots, T_{k-1}$ which proves the lemma.

\end{proof}

\subsubsection{Soundness}
\label{sec:ugcfvssound}
As we can choose the soundness parameter $\gamma$ of the Unique Games Conjecture to be arbitrarily small the following lemma says that, in the soundness case, there is no feedback vertex set containing less than a $(1-8 \delta)$ fraction of the test-vertices (or equivalently, any induced subgraph containing a $8\delta$ fraction of the test-vertices contains a cycle).
\begin{lemma}
\label{lemma:ugcfvssound}
If the graph has a FVS containing less than a $(1-8\delta)$ fraction of the test-vertices, then the Unique Game instance has a labeling that satisfies at least a fraction $\frac{\delta \eta^2}{t^2 k^2}$ of the constraints.
\end{lemma}
\begin{proof}
Consider a topological sort $\sigma: V \mapsto [n]$ of the graph obtained by deleting a FVS and
assume that it contains at least a $8\delta$ fraction of the
test-vertices, i.e., if we let $T$ be the set of remaining
test-vertices then
$$
\Pr_{x,S,v,w_1,\dots, w_{2t}}[ t_{x,S,v,w_1, \dots, w_{2t}} \in T] \geq 8 \delta.
$$ We shall show that this implies that there is a labeling of $\mathcal{L}$ that satisfies at least a $\frac{\delta \eta^2}{t^2 k^2}$ fraction of the constraints.

With each $w\in W$, we associate the
indicator function $f_{w} : [k]^R \mapsto \{0,1\}$ that takes value $0$ for the first half of
the bit-vertices corresponding to $w$ (according to $\sigma$)
and value $1$ for the remaining half.  We then define the set $L[w]$ of candidate labels for every $w\in W$ as:
$$
L[w] := \{i \in [R]: \Inf_i^d(f_{w}) \geq \eta\}.
$$ 
By Observation~\ref{obs:lowdeg}, we have $|L[w]| \leq d/\eta$.

Now, for every $w\in W$, we define $\rho(w)$ to be a random label from
$L[w]$ (if $L[w]$ is empty we assign any label to $w$) and, for every
$v\in V$ we pick a random neighbor $w\in N(v)$ and define $\rho(v) =
\pi_{v,w}( \rho(w))$. We shall now calculate a lower bound on the
expected number of edges the labeling $\rho$ satisfies.

We call a tuple $(v, w_1, \dots, w_{2t})$ good if
$$
\Pr_{x,S}[ t_{x,S, v, w_1, \dots, w_{2t}} \in T] \geq 4 \delta.
$$ Since $|T|$ contains a $8\delta$ fraction of the test-vertices we
have that at least $4\delta$ of the tuples are good.  Consider such a
good tuple $(v, w_1, \dots, w_{2t})$ and let $T_{v, w_1 \dots, w_{2t}}
= \{t_{x,S, v, \dots, w_{2t}} \in T\}$.  Suppose w.l.o.g. that
$$
\max_{f_{w_i}(x) = 0} \sigma(b_{w_i,x})  < \max_{f_{w_{i+1}}(x) = 0} \sigma(b_{w_{i+1},x})\qquad \mbox{for } i = 1, 2, \dots, 2t-1.
$$
Then for a test-vertex $t_{x,S, v,w_1, \dots, w_{2t} }$ to be in  $T_{v, w_1, \dots, w_{2t}}$ we must because of the arcs have
$$
f_{w_{t+1}}(C_{x,S,v,w_{t+1}}) \equiv \dots \equiv f_{w_{2t}}(C_{x,S,v,w_{2t}}) \equiv 0
$$
if $\sigma(t_{x,S}) \leq \max_{f_{w_{t}}(x) = 0} \sigma(b_{w_t,x})$
and otherwise, if $\sigma(t_{x,S}) > \max_{f_{w_{t}}(x) = 0} \sigma(b_{w_t,x})$,
$$
f_{w_{1}}(C^\oplus_{x,S,v,w_1}) \equiv \dots = f_{w_{t}}(C^\oplus_{x,S,v,w_t}) \equiv 1.
$$ Therefore, one of these conditions must be satisfied by at least
half of the test-vertices in $T_{x,S,v, w_1, \dots, w_{2t}}$, i.e., either
$$
\Pr_{x,S}\left[ \bigwedge_{j=t+1}^{2t} f_{w_{j}}(C_{x,S,v,w_j}) \equiv 0\right] \geq 2\delta
$$
or
$$
\Pr_{x,S}\left[ \bigwedge_{j=1}^{t} f_{w_{j}}(C^{\oplus}_{x,S,v,w_j}) \equiv 1\right] = \Pr_{x,S}\left[ \bigwedge_{j=1}^{t} (1-f_{w_{j}})(C_{x,S,v,w_j}) \equiv 1\right]\geq 2\delta.
$$
In either case,
Theorem~\ref{thm:soundmultiple} implies that there exist  $j\in
L[w_{\ell_1}]$ and $j'\in L[w_{\ell_2}]$ for some $1\leq \ell_1 \neq \ell_2 \leq 2t$ such that
$\pi_{v,w_{\ell_1}}(j) = \pi_{v,w_{\ell_2}}(j')$.

We now follow the same argumentations as used in~\cite{BK09}.
Overall, if we pick the tuple $(v, w_1, w_2, \dots, w_{2t})$ at random
and then $w,w'$ at random from the set $\{w_1, \dots, w_{2t}\}$, then
with probability at least $4\delta$ the tuple is good, with probability
$1/(4t^2)$ we have $w= w_{\ell_1}$ and $w'=w_{\ell_2}$, and with
probability $1/(k^2/\eta^2)$, the labeling procedure defines $j =
\rho(w), j' = \rho(w')$. Hence
$$
\Pr_{v, w, w'}[ \pi_{v,w}(\rho(w)) = \pi_{v,w'}(\rho(w'))] \geq \frac{4\delta \eta^2}{4t^2 k^2},
$$
and (expected over the randomness of the labeling procedure)
$$
\Pr_{(v,w)}[ \rho(v) = \pi_{v,w}( \rho(w))] \geq   \frac{\delta \eta^2}{t^2 k^2}.
$$
This shows that there exists a $\rho$ that satisfies a fraction $\frac{\delta \eta^2}{t^2 k^2}$ of the constraints.
\end{proof}

\subsection{Dag Vertex Deletion Problem}
\label{sec:ugcdvdp}
We prove the following theorem which clearly implies the DVD hardness stated in Theorem~\ref{thm:vertmain}.
\begin{theorem}
\label{thm:ugcdvdp}
Assuming the Unique Games Conjecture, for any integer $k \geq 2$ and arbitrary constants  $\epsilon, \delta>0$, given a directed graph $G(V,E)$, distinguishing between the following cases is \textsc{NP}-hard:
 \begin{itemize}
    \item \emph{(Completeness):} there exist disjoint subsets $V_1, \dots, V_k \subset V$ satisfying $|V_i| \geq \frac{1-2\epsilon}{k} |V|$ and such that a subgraph induced by all but one of these subsets has no path of length $k$.
        \item  \emph{(Soundness):} every induced subgraph of $32\delta |V|$ vertices contain a path of length $|V|^{1-\delta}$.
 \end{itemize}
\end{theorem}

We first present the reduction in the following subsection followed by the completeness (Lemma~\ref{lemma:ugcdvdpcomp}) and soundness (Lemma~\ref{lemma:ugcdvdpsound}) analyses.

\subsubsection{Reduction}

We  describe a reduction from Unique Games to DVD. Let $\mathcal{L}(G(V,W,E), [R], \{\pi_{v,w}\}_{(v,w)\in E}$ be  a Unique Games instance. As before, it will be convenient to present the DVD instance as it consists of two types of vertices that we refer to as \emph{bit-vertices} and \emph{test-vertices} and all edges are between bit- and test-vertices.

\begin{itemize}
\item The bit-vertices are partitioned into $L+1$ bit-layers ($L$ is selected below). Each bit-layer $\ell=0, \dots, L$ contains a bit-vertex $b^{\ell}_{w,x}$ of weight $\infty$ for every $w \in W$ and $x\in [k]^R$.

    In other words, each $w\in W$ is replaced by a $Q$-ary hypercube $[k]^R$ in each layer.

\item Similarly, the test-vertices are partitioned into $L$
  test-layers.  Each test-layer $\ell = 0, \dots, L-1$ has a
  test-vertex $t^\ell_{x,S, v, w_1, \dots, w_{2t}}$ of weight $1$ for every $x\in [k]^{R}$, every sequence of indices $S= (i_1, \dots, i_{\epsilon R}) \in [R]^{\epsilon R}$, every $v\in V$ and every sequence $(w_1, \dots, w_{2t}) \in N(v)^{2t}$ of (not necessarily distinct) $2t$ neighbors of $v$.

\item     The arcs incident to a test-vertex $t^{\ell'}_{ x,S, v,w_1, \dots, w_{2t}}$ are the following. For $j = 1, \dots, 2t$,
    \begin{itemize}
    \item there is an arc $(b^\ell_{w_j,z}, t^{\ell'}_{x,S, v, w_1, \dots, w_{2t}})$ if $\ell\leq \ell'$ and $z\in C_{x,S, v, w_j}$,
    \item and  an arc $(t^{\ell'}_{x,S, v, w_1, \dots, w_{2t}}, b^\ell_{w_j,z})$ if $\ell>\ell'$ and $z\in C^\oplus_{x,S, v, w_j}$.
    \end{itemize}
\item   Finally, $L$ is selected so as $\delta^2 L \geq |T|^{1-\delta}$ where $T$ is the set of test-vertices.
\end{itemize}

Similar to before, we can obtain an unweighted graph $G$  (with the set $T$ of test-vertices as vertices) with the
same optimal value by having an arc between two test-vertices if there
is a path between them through one bit-vertex.
Theorem~\ref{thm:ugcdvdp} therefore follows from proving that (i) we can partition the test-vertices as required in the completeness case (Lemma~\ref{lemma:ugcdvdpcomp}) and  (ii) that we have to delete almost all test-vertices in the soundness case (Lemma~\ref{lemma:ugcdvdpsound}) in order to avoid long paths.

\subsubsection{Completeness}

We show the following.

\begin{lemma}
\label{lemma:ugcdvdpcomp}
If there is a labeling $\rho$ of the Unique Games
instance $\mathcal{L}$ satisfying a $1-\zeta$ fraction of the constraints then we can partition the test-vertices into subsets $T', T_0, \dots, T_k$ satisfying
$T_j \geq \frac{1-2\epsilon }{k}
   |T|, |T'| \leq 2\epsilon  |T|$, and for $j\in [k]$ the graph obtained by deleting  $T' \cup T_j$  has no path of length $k$.
\end{lemma}
\begin{proof}
Note that if we collapse all layers by identifying the different copies of a bit-vertex and test-vertex in different layers then the DVD instance is equivalent to the FVS instance constructed in Section~\ref{sec:ugcfvs}. We can therefore (by the arguments of Section~\ref{sec:ugcfvscomp}), partition the bit-vertices into $k$ equal sized sets $B_0, B_1, \dots, B_{k-1}$ and all but an $2\epsilon $ fraction of the test-vertices into $k$ equal sized sets $T_0, T_1, \dots, T_{k-1}$ so that any test-vertex in $T_j$ has only incoming arcs from bit-vertices in $B_j$ and outgoing  arcs to bit-vertices in $B_{j\oplus 1}$.

Any $j\in [k]$ therefore corresponds to a solution by removing an $2 \epsilon $ fraction of the test-vertices (i.e., the set $T'$) and those test-vertices in $T_j$.
\end{proof}
\subsubsection{Soundness}

As we can choose the soundness parameter $\gamma$ of the Unique Games Conjecture to be arbitrarily small and $\delta^2 L \geq |T|^{1-\delta}$, the following lemma implies the soundness case of Theorem~\ref{thm:ugcdvdp}.
\begin{lemma}
\label{lemma:ugcdvdpsound}
If the Unique Game instance has no labeling that satisfies a fraction $\frac{\delta \eta^2}{t^2 k^2}$ of the constraints then every induced subgraph of the bit-vertices and  $32\delta |T|$ test-vertices has a path of length $\delta^2 |L|$.
\end{lemma}

\begin{proof}

As in Section~\ref{sec:dvdpsound}, it shall be convenient to look at the equivalent formulation of the problem where we wish to find a coloring $\chi$ that assigns a color in $\{1,2,\dots, k\}$ to each
bit-vertex satisfying $\chi(b^\ell_{w,x}) \leq \chi(b^{\ell'}_{w,x})$ if $\ell\leq \ell'$ so as to minimize the number of unsatisfied test-vertices
where a test-vertex $t^\ell_{x,S,v, w_1, \dots, w_{2t}}$ is said to be satisfied if
$$
\max_{\substack{1\leq j \leq 2t\\[1mm] z\in C_{x,S,v,w_j}}} \chi(b^{\ell}_{w_j,z}) < \min_{\substack{1\leq j \leq 2t  \\[1mm] z\in C^\oplus_{x,S,v, w_j}}} \chi(b^{\ell+1}_{w_j, z}),
$$
that is, all its predecessors are assigned lower colors than its successors.

We also generalize the concept of a  lower bound on the colors assigned to most bit-vertices corresponding to $w\in W$ in each layer:
define the color $\chi(w,\ell)$ of $w\in W$ and a bit-layer $\ell = 0, 1,
 \dots, L$ as the maximum color that satisfies $ \Pr_{x}[ \chi(b^\ell_{w,x})
   \geq \chi(w,\ell)] \geq 1-\delta. $

Now, with $w\in W$ and each test-layer $\ell= 0, 1, \dots, L-1$ we associate the
indicator function $f^\ell_w:[k]^R \mapsto \{0,1\}$ defined as follows
$$
f^\ell_w(x) =
\begin{cases}
0 & \mbox{if} \qquad \chi(b^{\ell+1}_{w,x})  > \chi(w,\ell),  \\
1 & \mbox{otherwise}.
\end{cases}
$$

Analogous to Claim~\ref{claim:keydictdvdp}, the key statement for the soundness analysis is the following.
\begin{claim}
Assuming that the Unique Games instance $\mathcal{L}$ has no labeling satisfying a fraction $\frac{\delta \eta^2}{t^2 k^2}$ of the constraints: if a fraction $16\delta$ of the test-vertices of test-layer
$\ell$ are satisfied, then  $\chi(w,\ell) < \chi(w,\ell+1)$ for at least a fraction $2\delta$ of the vertices in $W$.
\end{claim}
\begin{proof}
If we let $T$ be the set of satisfied test-vertices of test-layer $\ell$ then (as $T$ contains at least a  fraction $16 \delta$ of the test-vertices in that layer)
$$
\Pr_{x,S, v, w_1, \dots, w_{2t}}\left[
\max_{\substack{1\leq j \leq 2t\\[1mm] z\in C_{x,S,v,w_j}}} \chi(b^{\ell}_{w_j,z}) < \min_{\substack{1\leq j \leq 2t  \\[1mm] z\in C^\oplus_{x,S,v, w_j}}} \chi(b^{\ell+1}_{w_j, z})\right]
 \geq 16\delta.
$$
Similar to Section~\ref{sec:ugcfvssound}, we call a tuple $(v, w_1, \dots, w_{2t})$ good if
$$
\Pr_{x,S}\left[
\max_{\substack{1\leq j \leq 2t\\[1mm] z\in C_{x,S,v,w_j}}} \chi(b^{\ell}_{w_j,z}) < \min_{\substack{1\leq j \leq 2t  \\[1mm] z\in C^\oplus_{x,S,v, w_j}}} \chi(b^{\ell+1}_{w_j, z})\right]
 \geq 8\delta
$$ and note that at least a $8\delta$ fraction of the tuples are good.

By the definition of $\chi(w,\ell)$ we have  $\Pr_x[\chi(b_{w,x}^\ell) \geq
  \chi(w,\ell)] \geq 1-\delta$ and therefore for a good tuple $(v,w_1, \dots, w_{2t})$,
$$
7\delta \leq \Pr_{x,S}\left[
\max_{1\leq j \leq 2t} \chi(w_j,\ell) < \min_{\substack{1\leq j \leq 2t  \\[1mm] z\in C^\oplus_{x,S,v, w_j}}} \chi(b^{\ell+1}_{w_j, z})\right]
 \leq \Pr_{x,S}\left[
\bigwedge_{j=1}^{2t} f^\ell_{w_j} (C_{x,S,v, w_j}) \equiv 0 \right], 
$$
which, by Theorem~\ref{thm:soundmultiple}, implies that either 
\begin{itemize}
\item[(i)] more than $t$ of the functions are such that $\E[f^\ell_{w_j}] < \delta$ and hence 
$\chi(w_j, \ell+1) > \chi(w_j, \ell)$; or 
\item[(ii)] there exists $1 \leq \ell_1 \neq \ell_2 \leq t$ and $j \in L[w_{\ell_1}], j' \in  L[w_{\ell_2}]$ 
such that $\pi_{v,w_{\ell_1}}(j) = \pi_{v,w_{\ell_2}}(j')$, where (similar to Section~\ref{sec:ugcfvssound})
$$
L[w] = \{i\in [R]: \Inf_i^d(f^\ell_w) \geq \eta\}.
$$

\end{itemize}
If condition~(i) holds for half the good tuples, i.e., a fraction $4\delta$ of all tuples, then the statement follows because we can pick a vertex $w$ in $W$ uniformly at random
by first picking  a tuple $(v,w_1, \dots, w_{2t})$ at random and then picking one of the vertices $w_1, \dots, w_{2t}$ at random. With
probability $4\delta$ the tuple is good and satisfy condition~(i) and (conditioned upon that fact) with probability $1/2$ the picked vertex $w_j$ will be such that
$\chi(w_j, \ell+1) > \chi(w_j, \ell)$. Therefore, we have that if condition~(i) holds for half the good tuples then $\Pr_w[\chi(w,\ell+1) > \chi(w,\ell)] \geq 2\delta$ as required.

On the other hand, we shall show that the assumption of the claim (that no labeling of the Unique Games instance satisfies a fraction $\frac{\delta \eta^2}{t^2 k^2}$ of the constraints)
is violated if condition~(ii) holds for more than half the good tuples.
This follows from very similar arguments as in Section~\ref{sec:ugcfvssound} (and in~\cite{BK09}). Indeed, 
for every $w\in W$, define $\rho(w)$ to be a random label
from $L[w]$ and, for every $v\in V$  pick a random neighbor $w\in
N(v)$ and define $\rho(v) = \pi_{v,w}( \rho(w))$. 
If condition~(ii) holds for half the good tuples, then with probability $4\delta$ a random tuple $(v,w_1, \dots, w_{2t})$ is such a tuple, with probability $1/(4t^2)$ we have
we have $w= w_{\ell_1}$ and $w' = w_{\ell_2}$ for $w,w'$ randomly picked in the set $\{w_1, \dots, w_{2t}\}$, and with probability $1/(k^2/\eta^2)$, the labeling procedure defines $j = \rho(w), j' = \rho(w')$. Hence (if condition~(ii) holds for half the good tuples) 
$$
\Pr_{v, w, w'}[ \pi_{v,w}(\rho(w)) = \pi_{v,w'}(\rho(w'))] \geq \frac{4\delta \eta^2}{4t^2 k^2},
$$
and (expected over the randomness of the labeling procedure)
$$
\Pr_{(v,w)}[ \rho(v) = \pi_{v,w}( \rho(w))] \geq   \frac{\delta \eta^2}{t^2 k^2}.
$$
This shows that condition~(ii) cannot hold for half the good tuples as this would imply that there is a labeling 
of $\mathcal{L}$    that satisfies a fraction $\frac{\delta \eta^2}{t^2 k^2}$ of the constraints.

\end{proof}

To see how the above claim implies the lemma consider a subgraph induced by all bit-vertices and a fraction $32\delta$ of the test-vertices and consider the smallest
number of colors needed for a coloring $\chi$ to satisfy all those test-vertices.

Note that at least a fraction $16\delta$ of the test-layers are such that at least a fraction $16\delta$ of the test-vertices of that layer are satisfied by $\chi$.
This in turn, by the preceding claim, implies that  
$$\Pr_{\ell\in [L], w\in W}[ \chi(w,\ell+1) > \chi(w, \ell)] \geq 16\delta\cdot 2\delta = 32\delta^2
$$ 
and hence there exists a $w\in W$ such that $\Pr_{\ell\in [L], w\in W}[ \chi(w,\ell+1) > \chi(w, \ell)] \geq 32\delta^2$.
Therefore, the coloring $\chi$ needs to use at least $32\delta^2 L$ colors to satisfy a fraction $32\delta$ of the test-vertices or, in other words, 
any subgraph induced by the bit-vertices and a fraction $32\delta$ of the test-vertices has a path of length $32\delta^2L -1 \geq \delta^2 L$.
\end{proof}

\section{Discrete Time-Cost Tradeoff Problem}
\label{sec:timecost}

In the discrete time-cost tradeoff problem we are given a set $J$ of activities together with a partial order $(J,<)$. Any execution of the activities must comply with the partial order, that is, if $j<k$ activity $k$ may not be started until $j$ is completed. The duration of an activity depends on how much resources that are spent on it. This tradeoff between time and cost for each job is described by a nonnegative cost function $c_j: \mathbb{R}_+ \mapsto \mathbb{R}_+ \cup \{\infty\}$, where $c_j(x_j)$ denotes the cost to run $j$ with duration $x_j$. The project duration $t(x)$ of the realization $x$ is the makespan (length) of the schedule which starts each activity at the earliest point in time obeying the precedence constraints and durations $x_j$. Given a deadline $T>0$, the \emph{Deadline problem} is that of finding the cheapest realization $x$ that obeys the deadline, i.e., $t(x) \leq T$.

\begin{theorem}
The Deadline problem is as hard to approximate as DVD.
\end{theorem}
\begin{proof}
We reduce (in polynomial time) the problem of approximating DVD to that of
approximating the Deadline problem.  Given an instance of DVD, i.e., an integer $k$ and a DAG $G(V,A)$ with the
vertices ordered $0,1, \dots, n-1$ according to a topological sort,
consider the instance of the Deadline problem defined as follows:

\begin{itemize}
\item The deadline $T$ is set to $n$.
\item The set $J$ of activities contains three activities $l_i, m_i, r_i$ for each vertex $i \in V = \{0,1\dots, n-1\}$  with precedence constraints $l_i < c_i < r_i$ and cost functions
    $$
    c_{l_i}(x) = \begin{cases}
        0, & \mbox{if } x\geq i \\
        \infty, & \mbox{ otherwise }
    \end{cases} \qquad
        c_{m_i}(x) = \begin{cases}
        0, & \mbox{if } x\geq 9/10 \\
        1, & \mbox{ otherwise }
    \end{cases} \qquad
            c_{r_i}(x) = \begin{cases}
        0, & \mbox{if } x> n-1-i \\
        \infty, & \mbox{ otherwise }.
    \end{cases}
    $$

    In addition, there is an activity $a_{(i,j)}$ for each arc $(i,j) \in A$ with precedence constraints $m_i < a_{(i,j)} < m_j$ and cost function
    $$
        c_{a_{(i,j)}}(x) = \begin{cases}
            0 , & \mbox{ if } x \geq j-i - \frac{9}{10} + \frac{1}{10(k-1)}\\
            \infty, & \mbox{otherwise.}
        \end{cases}
    $$
\end{itemize}

See Figure~\ref{fig:projectsched} for an example of the Deadline problem corresponding to a DVD instance $G$ with $k=3$.
\begin{figure}[hbtp]
    \begin{center}
        \includegraphics[width=10cm]{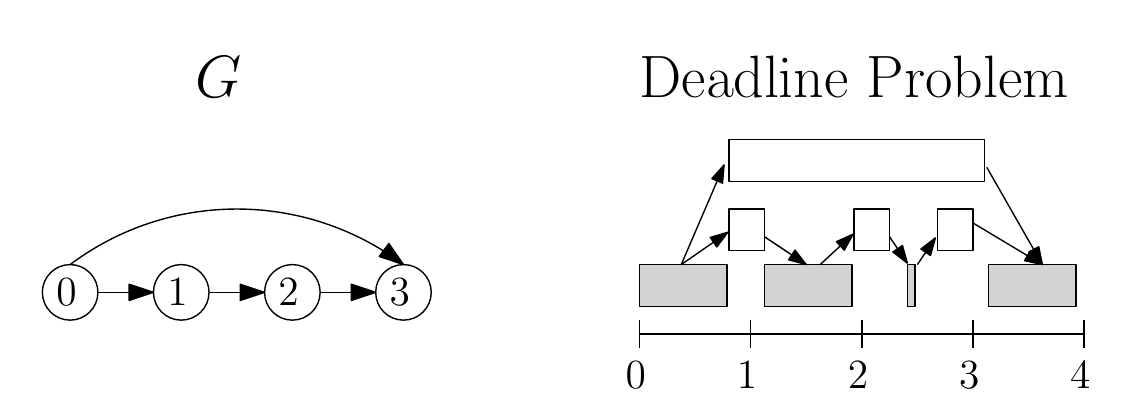}
    \end{center}
    \caption{For each vertex $i\in V$ the activity $m_i$ is depicted in light gray (activities $l_i$ and $r_i$ are omitted). The activities corresponding to arcs are depicted in white.
    Finally, the depicted solution pays a cost of $1$ for running activity $m_2$ in time $0$.}
    \label{fig:projectsched}
\end{figure}
Note that the cost functions of $l_i, m_i$, and $r_i$ enforces that activity $m_i$  has to be executed in the
interval $[i, i+1)$ and that it will require time $9/10$ unless we pay a cost of $1$ which allows us to run the activity in $0$ time.
Furthermore, as an activity $a_{(i,j)}$ always has duration (at least) $j-i-\frac{9}{10} + \frac{1}{10(k-1)}$, the start time $s_j$ of activity $m_j$ must
be such that $s_j - j \geq s_i-i + \frac{1}{10(k-1)}$, where $s_i$ is the start time of activity $i$.
 Using the fact that an activity $m_i$ must run in the interval
$[i,i+1)$ in order to obey the deadline, it follows that we have to pay a cost of $1$ for at least one activity corresponding to each path of length $k$.
By similar arguments, it also follows that this is also sufficient for having a realization respecting the deadline.
Therefore, any solution to the Deadline problem naturally corresponds to a solution to DVD (and vice versa) by deleting those vertices corresponding to
activities with a cost of $1$.

\end{proof}

\bibliographystyle{plain}
\bibliography{HardnessVertexDeletion}

\end{document}